\documentclass[aps,pre,twocolumn,nofootinbib,showpacs,showkeys]{revtex4-2}

\usepackage{dcolumn}
\usepackage{bm}
\usepackage{amssymb, amsmath, amsthm, amsfonts}
\usepackage{graphicx}
\usepackage{color}
\usepackage{mathrsfs}
\usepackage{physics}
\usepackage{mathtools}
\usepackage{standalone}
\usepackage{hyperref}
\usepackage{pgfplots}
\usepackage{appendix}
\pgfplotsset{compat=1.8}
\usepackage{tikz}
\usetikzlibrary{arrows.meta}
\usetikzlibrary{backgrounds}

\usepackage{tikz,fullpage}
\usetikzlibrary{arrows,%
                petri,%
                topaths}%
\usepackage{tkz-berge}

\usepackage{subcaption}

\usepackage{amsthm}
\usepackage{amsmath}
\usepackage{amssymb}
\newtheorem{theorem}{Theorem}
\newtheorem{prop}{Proposition}
\newtheorem{corollary}{Corollary}
\newtheorem{definition}{Definition}
\newtheorem{lemma}{Lemma}

\theoremstyle{remark}
\newtheorem*{remark}{Remark}

\usepackage[UKenglish]{isodate}
\usepackage[UKenglish]{babel}

\begin{document}

\title{Topological constraints on self-organisation in locally interacting systems}\thanks{An interactive preprint: \cite{website} and video abstract: \href{https://www.youtube.com/watch?v=cGcY-ReeGDU}{youtu.be/cGcY-ReeGDU} accompany this paper. We thank Karl Friston for helpful comments and Mikalai Shevko for assistance with code debugging. DARS is supported by the Einstein Chair programme at the CUNY Graduate Center and the VERSES Research Lab.}

\author{Francesco Sacco}
\affiliation{Allen Discovery Center at Tufts University, Medford, MA 02155}
\author{Dalton A R Sakthivadivel}
\email{dsakthivadivel@gc.cuny.edu}
\affiliation{Department of Mathematics, CUNY Graduate Center, New York, NY 10016}
\author{Michael Levin}
\affiliation{Allen Discovery Center at Tufts University, Medford, MA 02155}
\affiliation{Department of Biology, Tufts University}
\affiliation{Wyss Institute for Biologically Inspired Engineering, Harvard University}

\date{\today}

\begin{abstract}

All intelligence is collective intelligence, in the sense that it is made of parts which must align with respect to system-level goals. Understanding the dynamics which facilitate or limit navigation of problem spaces by aligned parts thus impacts many fields ranging across life sciences and engineering. To that end, consider a system on the vertices of a planar graph, with pairwise interactions prescribed by the edges of the graph. Such systems can sometimes exhibit long-range order, distinguishing one phase of macroscopic behaviour from another. In networks of interacting systems we may view spontaneous ordering as a form of self-organisation, modelling neural and basal forms of cognition. Here, we discuss necessary conditions on the topology of the graph for an ordered phase to exist, with an eye towards finding constraints on the ability of a system with local interactions to maintain an ordered target state. By studying the scaling of free energy under the formation of domain walls in three model systems---the Potts model, autoregressive models, and hierarchical networks---we show how the combinatorics of interactions on a graph prevent or allow spontaneous ordering. As an application we are able to analyse why multiscale systems like those prevalent in biology are capable of organising into complex patterns, whereas rudimentary language models are challenged by long sequences of outputs.

\end{abstract}


\maketitle

\section{Introduction}

Self-organisation is a fascinating phenomenon observed across diverse systems in nature and technology. In biology, self-organisation of complex structures and functions, from subcellular machinery to multicellular morphogenesis, requires alignment of parts in order to navigate a problem space toward specific adaptive ends  \cite{couzin2007collective, watson2011optimization, levin2019computational, kriegman2020scalable, watson2022design, baluvska2023cellular, miller2023revised}. Dynamics of this form can be viewed as a sort of autopoietic cognition, raising interesting questions about how self-organising systems navigate through configuration space given some target morphology \cite{stone1997spirit, deisboeck2009collective, abzhanov2017old, reber2021cognition, lyon2021reframing, levin2022technological, pio2023scaling, mathews2023cellular, levin2023bioelectric, levincollective, lagasse2023future, zhang2025classical}. More recently, current approaches to machine learning have involved physics-inspired models such as the Hopfield network \cite{hopfield1982neural,krotov2016dense,demircigil2017model} and spin glasses \cite{castellani2005spin}, and energy-based \cite{ranzato2006efficient} or diffusion models \cite{yang2023diffusion}. In these systems, in order to maintain a pattern out of equilibrium or produce sensible data over long time spans, there must exist a `condensed' phase with long-range order. Whilst some systems---like those prevalent in biology---demonstrate remarkable abilities to organise at large-scales, other systems that also exhibit some degree of collective intelligence---such as natural language models---have much more limited capabilities \cite{khandelwal-etal-2018-sharp, sun-etal-2021-long, zhao2021calibrate, malkin-etal-2022-coherence, kwa2025measuring}. Despite this, spin systems have also been used to model cellular morphogenesis \cite{graner1992simulation, chaturvedi2005multiscale, torquato2011toward, szabo2013cellular, weber2016cellular}, suggesting the difference goes beyond the substrate of intelligence considered. This paper is motivated by the following question: what is the functional distinction between simple language models and multicellular organisms, and can generative AI harness that property to achieve long-range order? In particular, we are interested in placing explicit thermodynamic bounds on the probability of some example systems occupying such a phase, with applications to estimating their capability to self-organise.

The question of how probable ordered configurations are as the number of elements increases is equivalent to asking about the existence of a phase transition, where the ultimate energetic constraints on the existence of an ordered phase arise from the topology of the interactions between subunits of the system. A quick argument due initially to Landau--Lifshitz \cite[\textsection 149]{landau} shows why the one-dimensional Ising model, the prototypical spin system, has no ordered phase at any temperature; their study of scaling is the technique we will employ in this paper, so we will review it here. Recall that a thermodynamically favourable change in state is one which decreases the free energy
\[
\Delta F = \Delta E - T \Delta S.
\]
Suppose a domain wall of perimeter $P$ forms at an arbitrary location in the chain. Using an effective Hamiltonian, calculating the free energy of the system with zero domain walls and the system with one domain wall shows that the change in free energy scales like $-T\log P$, meaning that forming a domain wall is always thermodynamically favourable for sufficiently large $P$. As a result, long chains are unstable under thermal fluctuations and disorder always propagates through the chain. In the thermodynamic limit there is no spontaneous magnetisation. Ising proves this by finding the partition function and calculating the mean magnetisation analytically in his celebrated 1925 paper, but the advantage of the argument of Landau and Lifshitz is that it applies generically to phase transitions in many different sorts of one-dimensional systems (though the details prove somewhat subtle) and uses a simple scaling argument. The argument to the contrary in dimension greater than one is famously due to Peierls \cite{peierls1936ising}, where it is shown that the change in internal energy has a factor of $P$ balancing the change in entropy. This is found by computing the number of spins an `island' interacts with---which is highly dependent on the properties of the lattice. 

The main results in this paper concern when an ordered phase can exist in a self-organising system based on its interactions. We will focus on the scaling relationship between energy and entropy as it is determined by combinatorial topology. We begin in Section \ref{sec:topo} by presenting our argument for a large universality class of models, and deriving necessary conditions for the existence of an ordered phase in this class. In Section \ref{sec:1d} we analyse the one-dimensional Potts chain as a simple model system, recovering the argument given by Landau--Lifshitz in a more general setting. Section \ref{sec:ar} maps autoregressive models \cite{box1970time} onto this one-dimensional framework and proves their inability to maintain long-range order, with a discussion of transformers following that in Section \ref{transformers-section}. In Section \ref{cliques-section} we apply these results to multiscale systems to evaluate the possibility of order in that context, based on the hierarchies built into the topology of interaction. Finally, we discuss implications of these results and future directions in Sections \ref{limitations-section} and \ref{sec:conclusion}. We conclude that certain systems cannot remain ordered over long ranges for all time, and suggest this no-go theorem as a source of fitness pressure for biological phenomena like stigmergy and embodiment. The key insight of our work will be that topology is the critical factor differentiating these systems. Namely, whilst cells in the human body can coordinate and organize over vast scales, forming coherent tissues and organs, language models models struggle to maintain consistency beyond their limited context windows. This disparity stems directly from the underlying topology of interactions in these systems.

\section{The main argument}
\label{sec:topo}

Consider a $k$-vertex graph with $k$-by-$k$ adjacency matrix $G$ and an $n$-ary variable on each vertex. Denote each such variable as $s_i$, indexed by $i \in \{1, \ldots, k\}$. The interactions between any set of spins will be given by the Hamiltonian $H$.

\begin{definition}
A windowed Hamiltonian $H$ is a Hamiltonian in which spin-spin interactions are defined only within a finite window of interaction $\omega$.
\end{definition}

\begin{remark}
    Without loss of generality, we will assume that the lowest energy state of any windowed Hamiltonian is zero. We will moreover assume the existence of an upper bound to the highest energy of any possible window Hamiltonian, $E^\mathrm{max}_i \leqslant E^\mathrm{max}$.
\end{remark}

We call a sum of such effective Hamiltonians a {\it local Hamiltonian}. In particular,

\begin{definition}
A local Hamiltonian is a Hamiltonian which can be decomposed into a sum of several window Hamiltonians $H=\sum_u H_u$, all of which have the same window length $\omega$.
\end{definition}

For intuition's purposes, one may notice a local Hamiltonian with window size equal to one is a discrete Markov field.

\begin{figure}[h]
    \centering
    \includegraphics[width=\linewidth]{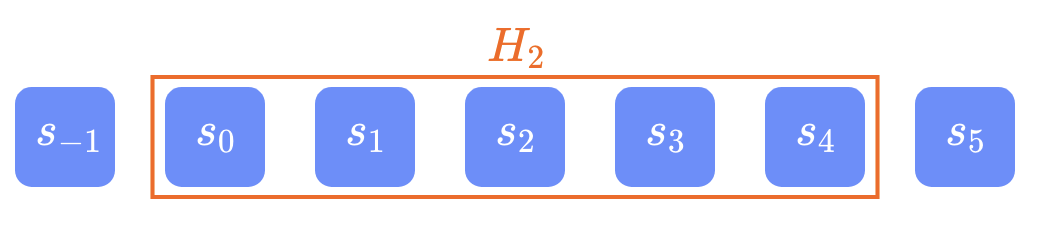}
    \caption{A local Hamiltonian in dimension one. The Hamiltonian of this chain is a sum of windowed Hamiltonians of length $\omega = 5$. Pictured is the second window; the first begins at $s_{-1}$.}
    \label{fig:windowed_ham}
\end{figure}

Our methodology will be given by the following prescription:
\begin{enumerate}
    \item Begin with the system in one of the ordered configurations, {\it e.g.} one of the stored patterns
    \item Create a domain wall
    \item Estimate the energy gained by the system
    \item Find the asymptotics of the free energy as the number of domain walls increases
\end{enumerate}

Computing the change in $F$ by changing the number of domain walls requires detailed knowledge of the combinatorics of the interactions on the graph. Instead we can use the structure imposed by the windows. Namely, we have regulated the length of interactions such that we need only count the number of windows containing a domain wall. This does away with the particularities of the lattice and so can be applied to systems on very generic graphs. We will explore this in the present subsection.

The most general way to consider the interaction between the elements of our system is if the interactions are represented by a graph Hamiltonian. For simplicity we will assume that the edges of the graph are fixed in time; one can easily justify this by appealing to an adiabatic approximation where the birth or death of edges is unlikely for the timescale on which observations are made.

\begin{definition}
Let $G$ be the adjacency matrix of a graph on $k$ vertices and $\bar{H}$ a $k$-by-$k$ matrix. A graph Hamiltonian is the Hamiltonian of a weighted directed graph; namely a Hamiltonian which can be written as
\[
    H=\bar{H} \odot G
\]
where $\odot$ is entry-wise multiplication and $\bar{H}$ has as entries the coupling strengths of the system.
\end{definition}

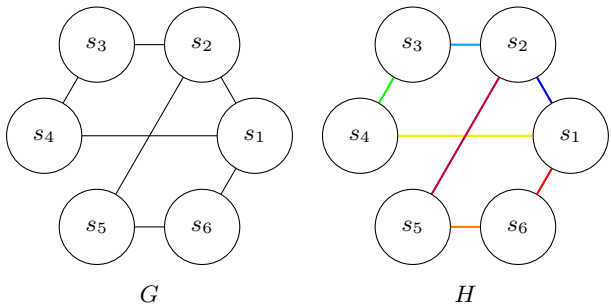
\begin{figure}[h]
\centering
\begin{tikzpicture}[scale=0.7]
\tikzset{vertex/.style = {draw, circle, fill=white, minimum size=1cm, inner sep=0pt}}
\tikzset{edge/.style = {-}}

\newcommand{\hexgraph}[2]{
    \foreach \i in {1,...,6}{
        \coordinate (#1\i) at ({60*(\i-1)}:2cm);
    }
    
    \draw[edge] (#11) -- (#12);
    \draw[edge] (#12) -- (#13);
    \draw[edge] (#13) -- (#14);
    \draw[edge] (#14) -- (#11);
    \draw[edge] (#15) -- (#16);
    \draw[edge] (#16) -- (#11);
    \draw[edge] (#12) -- (#15);
    
    \foreach \i in {1,...,6}{
        \node[vertex] (#1\i) at (#1\i) {$s_{\i}$};
    }
    
    \node at (0,-3) {#2};
}

\begin{scope}[shift={(-4,0)}]
    \hexgraph{A}{$G$}
\end{scope}

\begin{scope}[shift={(2,0)}]
    \hexgraph{B}{$H$}
    
    \draw[edge] (B1) -- (B2) node[midway, above right] {$e_1$};
    \draw[edge] (B2) -- (B3) node[midway, above] {$e_2$};
    \draw[edge] (B3) -- (B4) node[midway, above left] {$e_3$};
    \draw[edge] (B4) -- (B1) node[midway, above left] {$e_4$};
    \draw[edge] (B5) -- (B6) node[midway, below] {$e_6$};
    \draw[edge] (B6) -- (B1) node[midway, below right] {$e_7$};
    \draw[edge] (B2) -- (B5) node[midway, below right] {$e_5$};
\end{scope}
\end{tikzpicture}
\caption{Illustration of a graph Hamiltonian. Left: the graph $G$ represents the adjacency matrix of an undirected graph with six vertices, where edges indicate connections between vertices $s_i$. Right: the graph $H$ represents the graph Hamiltonian, where the coupling strengths $e_i$ have come from $\bar{H}$.}
\end{figure}

We will now give a recipe for computing the scaling behaviour of energy and entropy for a graph Hamiltonian system. For the sake of simplicity we will work with square planar embeddings of graphs (grid-like lattices) in the present paper. In these cases a domain wall and its perimeter can be easily defined (Figure \ref{fig:circle_grid}).

\begin{figure}[h]
    \centering
    \includegraphics[width=\linewidth]{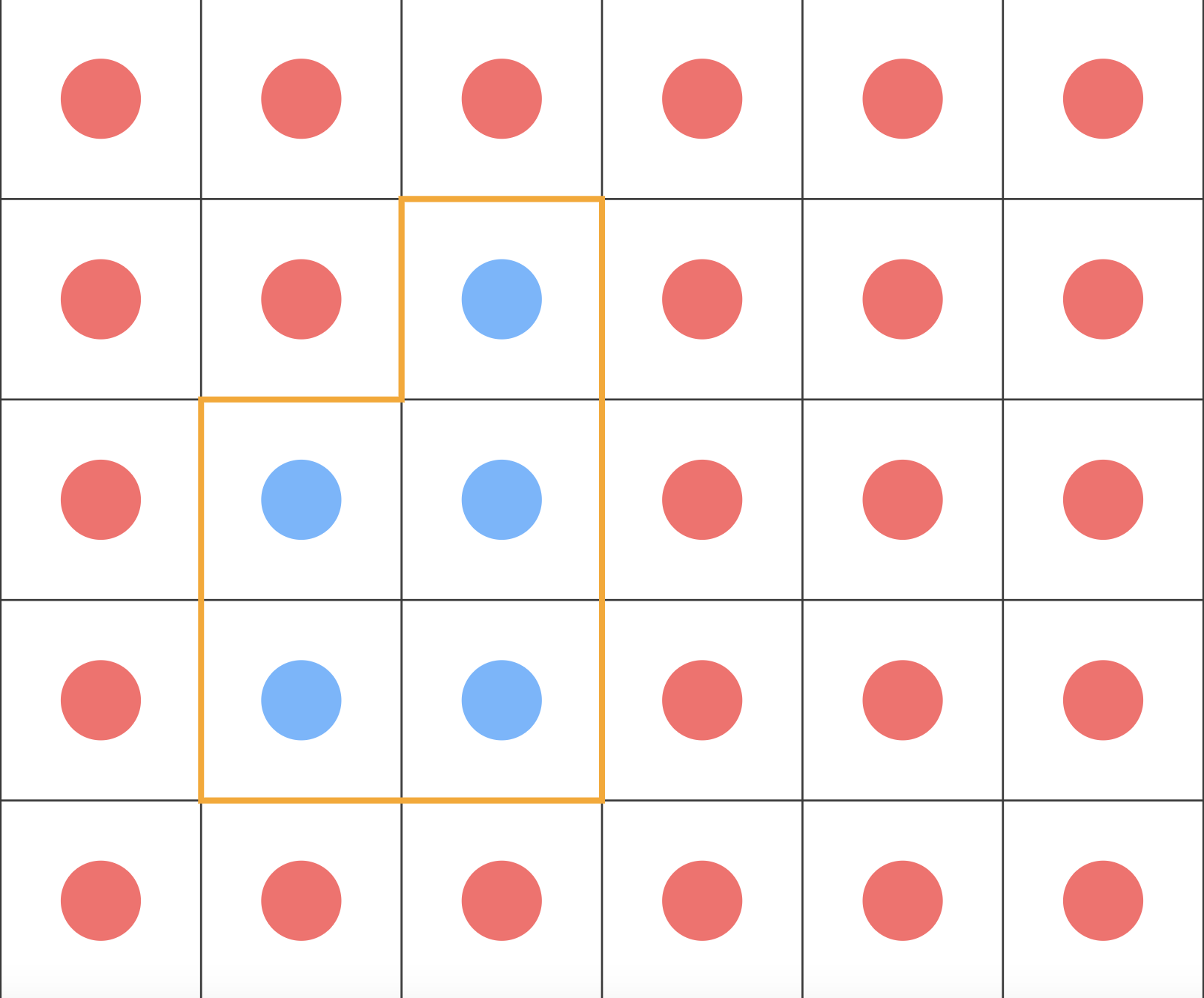}
    \caption{Two different domains in a two-dimensional grid. The perimeter that separates the two domains is drawn in orange.}
    \label{fig:circle_grid}
\end{figure}

Let $H$ be a graph Hamiltonian and $P$ be the perimeter length of a domain wall. As the perimeter length increases, the number of possible configurations of domain barriers increases, increasing the entropy of the system $\Delta S$. We say that the entropy gained scales as $f_S$ if
\begin{equation*}
    \Delta S=O(f_S(P)).
\end{equation*}
We give a similar definition for energy scaling: as the perimeter length increases, the upper and lower bounds of the energy gained scale as $O\left(f_E^\textrm{high}(P)\right)$ and $O\left(f_E^\textrm{low}(P)\right)$ respectively. If $f_E^\textrm{high}=f_E^\textrm{low}\equiv f_E$ we say that the energy gained scales as $f_E$, that is, $\Delta E=O(f_E(P))$.

If we can estimate how the energy and entropy scale, we can verify the existence or non-existence of an ordered phase without knowing the exact formulae for $E(P)$ and $S(P)$, only considering how they scale as $P\to \infty$.

\begin{prop}\label{prop:scaling}
If $O(f_S)=O(f_E)$ then there exists an ordered phase.
\end{prop}
\begin{proof}
Suppose $O(f)=O(f_S)=O(f_E)$. Then the change in free energy is
\begin{equation*}
\Delta F= \Delta E -T\Delta S=\lim_{P\to \infty}O(f(P))-TO(f(P))
\end{equation*}
If we now take $T$ to zero, the change in free energy becomes increasingly positive. As such, the creation of a domain wall is unfavourable.
\end{proof}

In this way we can rule out details of the system that could complicate the combinatorics---for example, the number of stored patterns. Let 
\[
\pi = (\ldots, s_{-1}, s_0, s_1, \ldots)
\]
denote a ground state of the system. One may think of this state as a stored pattern or morphogenetic configuration. When there is a ground state degeneracy---that is, multiple states with zero energy exist, for instance, if a Hopfield model has $m > 1$ stored patterns---they will be enumerated as 
\[
\pi^{\alpha} = (\dots,s^{\alpha}_{-1},s^{\alpha}_0,s^{\alpha}_1,\dots), \quad 1 < \alpha \leqslant m.
\]
We can show that the number of such patterns affects neither the scaling of the entropy nor energy.

\begin{lemma}\label{lm:energy_scaling}
Let $H=\sum_u H_u$. If there exist two energies $E^\mathrm{max}, E^\mathrm{min}$ which are the greatest and least non-zero energy levels of all the windowed Hamiltonians $H_u$ respectively, then at thermal equilibrium, the ability to converge to an ordered phase is independent of energy levels and window sizes.
\end{lemma}
\begin{proof}
For any $H_u$ let $\omega_1$ the size of the smallest window and $\omega_2$ that of the largest. The energy gained from the creation of a domain wall is bounded by
\begin{equation*}
\omega_1 P E^\textrm{min}\leqslant \Delta E\leqslant \omega_2 PE^\textrm{max}.
\end{equation*}
In both cases we have $E = O(P)$ so that the asymptotics of the system depend only on the perimeter length.
\end{proof}

\begin{lemma}\label{lm:entropy_scaling}
Let $H$ be a graph Hamiltonian with $m>1$ stored patterns. At thermal equilibrium, the ability to converge to an ordered phase is independent of $m$.
\end{lemma}
\begin{proof}
Let $B(P)$ be the number of possible configurations of domain barriers with perimeter $P$. The change in entropy due to the creation of a domain barrier can always be written as
\[
\Delta S= \log\left[(m-1)B(P)\right] =\log B(P) + \log (m-1).
\]
In the thermodynamic limit, the term proportional to the number of barriers increases, whilst the term proportional to the number of patterns stored stays constant. As such its contribution can be ignored. 
\end{proof}

Since neither the shape nor number of stored patterns affect the thermodynamics of the problem, and neither does the size of the interaction window, we reduce our analysis to the nearest-neighbour Ising model. Our main technical result is the following topological equivalence theorem:

\begin{theorem}\label{main-thm}
All local Hamiltonians on lattices with the same combinatorial structure have asymptotically equivalent free energies.
\end{theorem}
\begin{proof}
We will show we can approximate any $H$ with the Hamiltonian of some arbitrary other model in such a way that their asymptotics are equivalent, implying that in the thermodynamic limit, the Peierls argument depends only on the perimeter length in the approximate system. Let $H = \sum_u H_u$ be a local Hamiltonian on a planar graph $G$ with any spin $i$ coupled to $\eta_i$ other spins. Let $H_0$ be the Hamiltonian of an arbitrary system. One knows by the Bogoliubov inequality \cite{notes} that 
\[
F \leqslant F_0 + \langle H - H_0 \rangle_0
\]
and therefore that if $H$ and $H_0$ are asymptotically equivalent then $F \sim F_0$. As the perimeter length increases we have $E = O(P)$ and $E_0 = O(P_0)$ by Lemma \ref{lm:energy_scaling}. If the combinatorics of the lattices are the same---namely, if the set of $\eta_i$ is equal on both lattices---then we can take $P = P_0$. The change in the energy is computed by the Hamiltonian of the new configuration, implying $H - H_0 \sim 0$. By Lemma \ref{lm:entropy_scaling}, the probability of any disorder occurring is asymptotically the same in both systems, so that the {\it change} in free energies is also asymptotically equivalent. The claim follows.
\end{proof}

\begin{corollary}
    The \emph{capacity} for self-organisation in any system on a graph $G$---that is, the existence or non-existence of a phase transition---is equivalent to that of a nearest-neighbour Ising model on the same graph, with
    \[
    H_0 = -J \sum_{i,j} s_i s_j
    \]
    and two stored patterns 
    \begin{gather*}
    \pi^1 = (1, 1, 1, \ldots, 1) \\ \pi^2 = (-1, -1, -1, \ldots, -1).
    \end{gather*}
\end{corollary}

An important remark is that we do not claim the phase transitions are the same. Naturally they will differ in general, for instance if there are more stored patterns in one system than the other. It is the existence of a phase transition from the topology of the lattice which we study.

\section{Stored patterns in the windowed Potts model}\label{sec:1d}

With this in hand, we can study our first model system: the one-dimensional Potts model. The Potts model is a variant of the Ising model whose state vectors are richer than simply binary numbers. In particular, one can encode data in a Potts chain by choosing different spin configurations, recalling the discussion in the Introduction. It is a suitable model to describe stored images or other configurations of multiply-valued variables. 

\begin{definition}
Let $n$ be a finite positive number. A Potts chain $\mathcal{C}$ is a $\mathbb{Z}$-indexed set of $n$-ary integer symbols; that is, a chain of spins $s_i$ indexed by $i \in \mathbb{Z}$, where each $s_i$ can assume any integer value in $\{1, \ldots, n\}$. 
\end{definition}

It follows from Theorem \ref{main-thm} that it is sufficient to consider the scaling of $\mathcal{C}$ to establish the existence of an ordered phase. In the same fashion of Landau--Lifshitz, it can be shown that local Hamiltonians in dimension one do not converge to a prescribed ordered state, since the formation of a domain wall always `interrupts' the pattern.

\begin{theorem}
\label{th:no1d_order}
Let $H$ be a one-dimensional local Hamiltonian with $m > 1$ stored patterns. At non-zero temperature the formation of a domain wall is thermodynamically favourable. 
\end{theorem}
\begin{proof}
Suppose that our Potts chain starts out in the first ground state or pattern, $\mathcal{C} = \pi^1$. If
\[
    \Delta F = \Delta E - T\Delta S < 0,
\]
so that the free energy of the system decreases upon the formation of a domain barrier, then the formation of a domain barrier is thermodynamically favourable. It is immediate that any $\pi^\alpha$ differs from some other $\pi^\gamma$ by at least one change in spin. In a sequence of length $L$ there are $L-1$ possible places where a domain wall can appear, and at each such place we may obtain one of the $m-1$ other patterns saved. As such the change in entropy of the system is
\[
\Delta S = \log [(m-1)(L-1)].
\]
Upon the formation of a domain barrier, the windowed Hamiltonians that intersect it will have non-zero, positive energy. By assumption the energy is bounded, and no more than $\omega$ windows can be affected by a domain wall, from which we obtain
\[
    0\leqslant \Delta E\leqslant \omega E^\textrm{max}.
\]
The change in free energy is 
\[
    \Delta F \leqslant \omega E^\textrm{max} - T\log [(m-1)(L-1)].
\]
Assume $T > 0$. As $L$ increases, the right hand side of the equation eventually becomes negative.
\end{proof}

\section{Autoregressive models}

\label{sec:ar}

We begin this section by recalling in Figure \ref{fig:text-genesis} the motivation stated in the introduction: contrasting the generation of textual features with the generation of morphological features, we would like to understand how the self-organising nature of text is constrained by the topology of the interactions between subunits generating that text.

\begin{figure}[h]
    \centering
    \includegraphics[width=\linewidth]{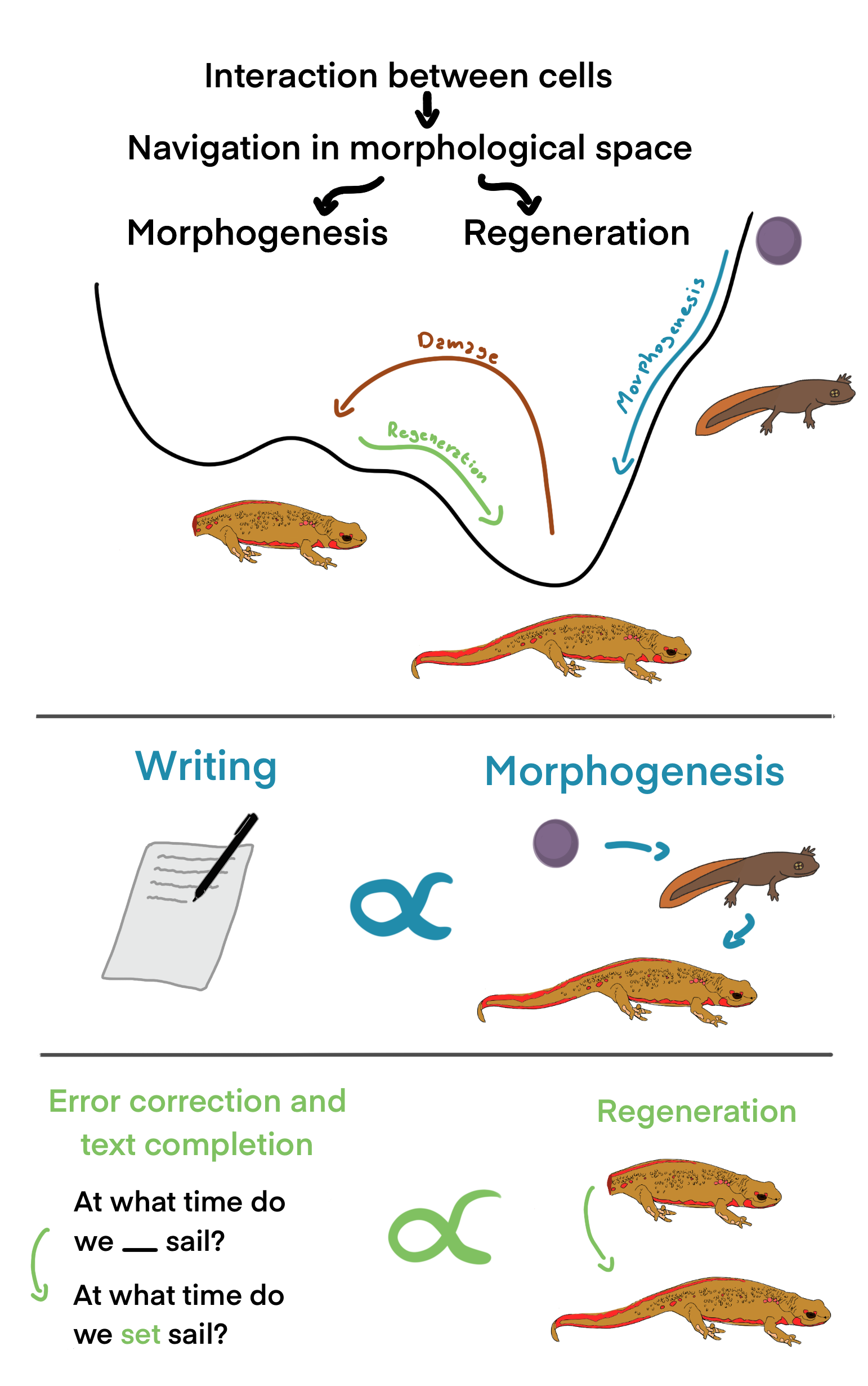}
    \caption{Text generation in analogy to a morphogenetic process.}
    \label{fig:text-genesis}
\end{figure}

We will now discuss a similar result as the one in the previous subsection, for autoregressive models. An autoregressive model is formed when values of a sequence are regressed against previous values of that sequence. Here it will be defined as an estimator of some conditional probability distribution; namely, that of a sequence of observations of some data at step $i$, given a history of observations until $i-1$.

\begin{definition}
Let $\{s_1,\dots, s_{i-1}\}$ be a random $n$-ary sequence of length $i-1$. Given a window (also called context) of length $\omega$, an order $\omega$ autoregressive model computes 
\[
P(s_i \mid s_{i-1}, \ldots s_{i - \omega})
\]
as a probability vector over $\{1, \ldots, n\}$ by generating samples of $s_i$ according to some random process on window states.
\end{definition}

Such a function is called an $\mathrm{AR}(\omega)$ model in particular. As shorthand we will denote the estimator whose output is $v$ when given a sequence $\{s_{i-1}, \ldots s_{i - \omega}\}$ as $M$: 
\[
v \coloneqq M(s_i \mid s_{i-1}, \ldots s_{i - \omega}).
\]
The function $M$ has the type of a conditional probability distribution. We further write
\[
P(s_i = c) = M(s_i = c \mid s_{i-1}, \ldots s_{i - \omega})
\]
for the $c$-th component of $v$.

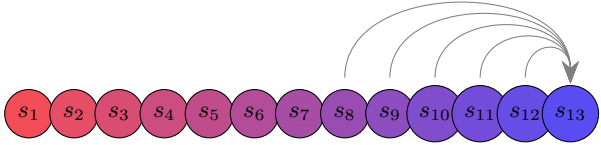
\begin{figure}[h]
    \centering
    
    \begin{tikzpicture}[scale=0.4]
        \colorlet{startcolor}{blue!70!white}
        \colorlet{endcolor}{red!70!white}
        
        \foreach \i in {1,...,13} {
            \pgfmathsetmacro{\x}{(\i-1)*1.5}
            \pgfmathsetmacro{\colorval}{100*\i/14}
            \colorlet{currentcolor}{startcolor!\colorval!endcolor}
            \node[circle, draw=black, fill=currentcolor, minimum size=0.4cm] (s\i) at (\x,0) {\small $s_{\i}$};
        }
        
        \foreach \i in {8,...,12} {
            \pgfmathsetmacro{\endpoint}{(\i-1)*1.5}
            \pgfmathsetmacro{\start}{12*1.5}
            \pgfmathsetmacro{\mid}{(\endpoint+\start)/2}
            \pgfmathsetmacro{\height}{3.5-0.5*(\i-10)}
            \draw[-{Stealth[length=3mm]}, gray] (\endpoint,1.2) .. controls (\endpoint,\height) and (\start,\height) .. (\start,1);
        }
        
    \end{tikzpicture}
    \caption{An autoregressive model with $\omega = 5$.}
    \label{fig:autoregressive_ising}
\end{figure}

To fit this in the discussion of patterns and long-range order in spin chains, we find the following theorem useful. We will set the convention that if $u - \omega < 1$ then the window is empty at that index, and that conditioning on the empty set is the same as taking unconditional probability.

\begin{theorem}
    \label{th:arlm}
    A unique local Hamiltonian with window length $\omega$ can be associated to any $\mathrm{AR}(\omega)$ model.
\end{theorem}
    \begin{proof}
    Let $M$ be our autoregressive model, and $s_{i-1},\dots,s_{i-\omega}$ our input sequence. The probability that the next observed spin in the sequence is equal to $c$ is
    \[
    P(s_i = c) = M(s_i = c \mid s_{i-1}, \ldots s_{i - \omega}).
    \]
    Suppose\footnote{One could appeal to Gibbs' equation or more general maximum entropy principles to do so.} we assign an energy to each possible $c$ with a scalar function $E\colon \{1, \ldots, n\} \to \mathbb{R}$ and constant weight $\beta \geqslant 0$, such that
    \[
    \beta E_c = -\log P(s_i = c) +\textrm{const} \quad\textrm{with}\quad c \in \{1\dots n\}.
    \]
    Set the constant in such a way that the lowest energy state has energy equal to zero. Now define a windowed Hamiltonian
    \[
        H_u(s_u) = -\log M(s_{u}|s_{u-1},\dots,s_{u - \omega})+\textrm{const}
    \]
    for any $u$ in $\{1, \ldots, i\}$ so that the probability of any $u$-th token occurring before $i$ is conditioned on the window preceding it. The full local Hamiltonian is the sum 
    \[
    H=\sum_u H_{u}(s_u)
    \]
    in an obvious way.
\end{proof} 

\begin{remark}
By Theorem \ref{th:arlm}, the generation of samples by an autoregressive model can now be seen as sampling from a Boltzmann distribution over sequences.
\end{remark}

We close this subsection with the following conclusion. The free energy can be calculated by $F = E - \beta^{-1} S$ as usual. By Theorem \ref{th:arlm}, autoregressive models are one-dimensional systems described by a local Hamiltonian. The following corollary is then a consequence of the scaling argument in Theorem \ref{th:no1d_order}. 

\begin{corollary}
    For any finite $\beta$, an autoregressive model is unable to converge to a single stored pattern. 
\end{corollary}

\section{Transformers and attention}\label{transformers-section}

An interesting question to ask is what these results imply for the organisational capabilities, and hence intelligence, of large language models. Most state-of-the-art language models of today are decoder-only transformer architectures (see {\it e.g.} \cite{lin2022survey}), meaning that to predict the next word (or token) they perform autoregression on all the previous elements of text---in other words, they are autoregressive at inference-time. Indeed, transformers can be described as spin collectives \cite{ramsauer2020hopfield, millidge2022universal, krotov2023new}, with training dynamics under asynchronous updates behaving like a spin chain evolving under Glauber dynamics and text generation being sampling from the corresponding equilibrium distribution, just as we have modelled autoregression by a spin chain here. A consequence of the results in \S\ref{sec:ar} is a fundamental limitation on the coherence of a simple version of a large language model for long sequences of outputs---and hence an explanation for limitations observed in papers such as \cite{kwa2025measuring}.

\begin{prop}
    Causally-masked attention in a decoder-only model has no ordered phase. 
\end{prop}
\begin{proof}
    Let $s \in \mathbb{R}^k$ be the configuration of the network, $X$ a matrix of $m$ stored patterns $(\pi^1, \ldots, \pi^m)$ with any $(X)_\alpha = \pi^\alpha = (s_1, \ldots, s_k)$ and $X_{i\alpha} = \pi^\alpha_i = s_i$, and $\beta$ a constant scalar. Now $X$ is a $k$-by-$m$ matrix of $n$-ary variables. One knows (see {\it e.g.} \cite{ramsauer2020hopfield}) the standard attention algorithm with the identity matrix for weights is a modern Hopfield network with an exponential potential function \cite{demircigil2017model}. The asynchronous updates maximise the probability of a Gibbs distribution, where the energy $\beta E(s)$ of a configuration, $-\log P(s) + \mathrm{const}$, is
    \begin{align*}
        - &\log(\sum_{\alpha=1}^m \exp(\beta \sum_{i=1}^k X_{i\alpha} s_i)) + \frac{\beta}{2}\sum_{i=1}^k s_is_i \\ &+ \log m + \frac{\beta}{2}\left(\max_\alpha \| X_\alpha\|\right)^2.
    \end{align*}
    With a causally-masked positional encoding from $i-\omega$ to $i$, the expression at $i$ with context length $\omega$ becomes
    \begin{align*}
        - &\log(\sum_{\alpha=1}^m \exp(\beta \sum_{j=i-\omega}^i X_{j\alpha} s_j)) + \frac{\beta}{2}\sum_{j=i-\omega}^i s_js_j \\ &+ \log m + \frac{\beta}{2}\left(\max_\alpha \| X_\alpha\|\right)^2.
    \end{align*}
    This can easily be seen as an expression for a local Hamiltonian in a particular $\mathrm{AR}(\omega)$ model; it follows from Theorem \ref{th:arlm} that we can apply Theorem \ref{th:no1d_order} to the transformer, completing the proof of the statement. 
\end{proof}

For simplicity we have used no mapping in an associative space, however the result generalises readily to those cases by the results in \cite{ramsauer2020hopfield}. 

More sophisticated large language models employ multi-headed attention algorithms \cite{vaswani2017attention}, where different attention heads in the decoder are attenuated to different classes of relationship. However, even in attention with heads for long-range dependencies, a maximum context length is imposed as a hyperparameter for reasons of computational cost (see {\it e.g.} the implementations in the review \cite{phuong2022formal}). Whilst the algorithm itself is more sophisticated, and allows better performance when applied to two-dimensional data ({\it e.g.} images), the coherence of a sequence of text is constrained by context length, which is a realistic and necessary-to-consider constraint given hardware. This limitation is shown empirically by the fact that the length of tasks which modern large language models are able to perform is still finite \cite{kwa2025measuring}, despite engineering pushing this limit further and further.

More discussion on this will be given at the end of \S\ref{cliques-section}.

\section{Large hierarchical systems on graphs}\label{cliques-section}

In many systems of interest there are hierarchical interactions, {\it i.e.} effective interactions between subgraphs. Many multiscale systems in biology and physics exhibit complex patterns consisting of pockets of regularity assembled into structures with global irregularity, such as tissues with different morphogenetic features assembled out of their constituent cells \cite{thomson1988morphogenesis, engler2009multiscale, levin2012morphogenetic, chan2017coordination, kuchling2020morphogenesis, mcmillen2024collective}, functional networks in the brain consisting of regions specialised to process certain sorts of information \cite{balduzzi2008integrated, kanwisher2010functional, johnson2011interactive, ramstead2021neural}, and self-assembling molecules in active matter situations \cite{sanchez2012spontaneous, haxton2013hierarchical, whitelam2015hierarchical, doncom2017dispersity, mcgivern2020active}. In this section we will investigate the interplay between the scaling of free energy at one level of a hierarchy and that of another, when the topology of the graph organises those levels meaningfully, to quantify when such patterns are possible. We will complement the results in previous sections by showing that systems lacking hierarchical structure may be limited in their ability to form complex patterns.

Recall that a clique is a (non-empty) complete induced subgraph. Suppose there exist $\ell > 1$ independent\footnote{By independent we mean the vertex set of each clique is disjoint from that of each other clique.} cliques in $G$, with $n_1, \ldots, n_\ell$ vertices respectively, and each $n_i > 2$. There will be a macrostate associated to each clique giving the magnetisation of that subgraph. We will argue that there are ways for local order but global disorder to exist in systems with cliques. Since the effective behaviour of the system, {\it i.e.} the properties of the cliques, is often a meaningful experimental variable (and hence a useful macrostate), this says organising a system hierarchically can create interesting order phenomena. In particular, we want to show there can exist `hierarchical behaviours', where each clique individually is a coherent phase, but that phase varies from clique to clique.

\begin{figure}[h]
\centering
\begin{tikzpicture}[scale=0.7]
\tikzset{vertex/.style = {draw, circle, fill=white, minimum size=1cm, inner sep=0pt}}
\tikzset{edge/.style = {-}}

\newcommand{\hexgraph}[2]{
    \foreach \i in {1,...,6}{
        \coordinate (#1\i) at ({60*(\i-1)}:2cm);
    }
    
    \draw[edge] (#11) -- (#12);
    \draw[edge] (#12) -- (#13);
    \draw[edge] (#12) -- (#14);
    \draw[edge] (#13) -- (#14);
    \draw[edge] (#14) -- (#16);
    \draw[edge] (#15) -- (#16);
    \draw[edge] (#16) -- (#11);
    \draw[edge] (#15) -- (#11);

    \foreach \i in {1,...,6}{
        \node[vertex] (#1\i) at (#1\i) {$s_{\i}$};
    }
    
    \node at (0,-3) {#2};
}

\begin{scope}[shift={(-4,0)}]
    \hexgraph{A}{$G$}
\end{scope}

\begin{scope}[shift={(2,0)}]
    \hexgraph{B}{Cliques in $G$ highlighted}
    
    \draw[edge, blue, thick] (B2) -- (B3);
    \draw[edge, blue, thick] (B3) -- (B4);
    \draw[edge, blue, thick] (B2) -- (B4);
    \draw[edge, red, thick] (B5) -- (B6);
    \draw[edge, red, thick] (B5) -- (B1);
    \draw[edge, red, thick] (B1) -- (B6);
    \end{scope}
\end{tikzpicture}
\caption{A graph $G$ with two independent $3$-vertex cliques and two edges connecting the cliques.}
\end{figure}
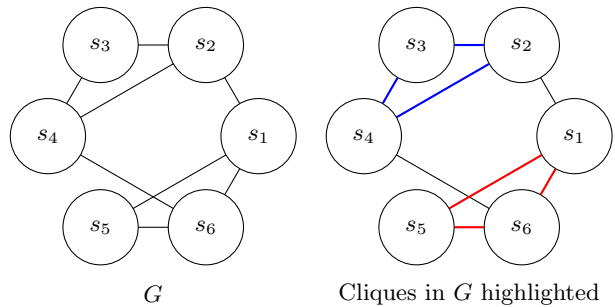

For brevity, and without loss of generality\footnote{This is because the $J$'s do not affect the combinatorics used to compute the entropy; as in Lemma \ref{lm:energy_scaling}, we are permitted to simply take $J$ large enough to bound all of the true couplings in the clique.}, we will assume the coupling constants $J$ are uniform across the graph. We will need the following observation before we prove our main theorem in this section. Recall that Boltzmann's formula $S = k \log W$ denotes by $W$ the multiplicity of a macrostate. A clique is positively (negatively) magnetised if all spins in the clique have value $+1$ ($-1$). We will begin with $\ell$ positively magnetised cliques. If the number (denote it $r$) of `flipped' ({\it i.e.} uniformly changed) cliques is a macrostate, then the multiplicity is the number of ways to arrange $r$ distinguished cliques out of the total $\ell$ cliques. As such, we have 
\[
S = k \log {\ell\choose r}.
\]

\begin{theorem}\label{clique-thm}
    Take any of the $\ell$ cliques. There exist parameter regimes where individual cliques may change from positive to negative magnetisation. For that temperature, take any individual $n_i$-clique and consider a spin within it. If the coupling of every spin in the clique is greater than $\frac{T}{2}\log n_i$ then the clique remains uniformly magnetised.
\end{theorem}
\begin{proof}
Flipping $r$ cliques from positive to negative magnetisation takes an energy of $\sum_{\gamma=1}^r 2J{n_i}_\gamma$ where $\gamma$ indexes cliques flipped, and has an entropy of the logarithm of $\ell$ choose $r$. The full change in free energy is 
\[
\Delta F = \sum_{\gamma=1}^r 2J{n_i}_\gamma - T\log\frac{\ell!}{r!(\ell - r)!}
\]
so that the scaling behaviour of $r$ flips is 
\[
\Delta F \sim O(r) - TO(r \log \ell - r \log r + r),
\]
where we have obtained the approximation of the entropy from ${\ell\choose r} \approx \ell^r / r!$ and Stirling's formula. Clearly if (up to some constant)
\[
T > \frac{1}{\log \ell - \log r + 1}
\]
then $\Delta F < 0$. Hence there exist temperatures for which it is favourable for some number $r$ of entire cliques to change. Fix such a $T$. Within an $n_i$-clique there are $n_i$ ways to flip a single spin. This takes energy $2J$. If we have $2J > T\log n_i$, then there is no domain wall within the clique.
\end{proof}

By minding constants in the $O(r)$ terms above we show the hypothesis can be attained---namely, that there exist values of $T$ for which the inequality is satisfied.

\begin{prop}
    Let $n_{\mathrm{max}}$ be an integer greater than zero denoting the number of vertices in the largest clique. There exists a non-empty critical temperature range of hierarchical behaviour.
\end{prop}
\begin{proof}
    From the theorem above we have 
    \[
    \frac{2J}{\log n_i} > T > \frac{2J \sum_{\gamma=1}^r {n_i}_\gamma }{r \log \ell - r \log r + r}.
    \]
    We also know that 
    \[
    r n_{\mathrm{max}} \geqslant \sum_{\gamma=1}^r {n_i}_\gamma
    \]
    and
    \[
    \frac{1}{\log n_{i}} \geqslant \frac{1}{\log n_{\mathrm{max}}},
    \]
    so it suffices to take 
    \[
    \frac{2J}{\log n_{\mathrm{max}}} > T > \frac{2J r n_{\mathrm{max}}}{r \log \ell - r \log r + r}.
    \]    
    It follows that
    \[
    \frac{1}{\log n_{\mathrm{max}}} > T > \frac{1}{\log \ell - \log r + 1} n_{\mathrm{max}}.
    \]
    Since the upper bound is a strictly decreasing function of $n_{\mathrm{max}}$ with a singularity at one, and the lower bound is a linear function of $n_{\mathrm{max}}$, they intersect for some sufficiently small slope, before which the inequality is satisfied. We conclude that there exist $\ell, r$ for which the set of satisfactory $T$ is non-empty. In particular, this occurs whenever
    \[
    \frac{\ell}{r} > \frac{n_{\mathrm{max}}^{n_{\mathrm{max}}}}{e}.
    \]
\end{proof}

By looking at the effective behaviour of the graph we consider the statistical properties of cliques to be like those of individual spins. When those cliques themselves form cliques, we have complete induced subgraphs on an effective graph and can make this argument again (see Figure \ref{multiscale-fig} for an example of what is meant).

\begin{figure}[h]
\centering

\begin{subfigure}{0.5\textwidth}
\begin{tikzpicture}[scale=0.65, transform shape]

\tikzset{vertex/.style = {draw, circle, fill=white, minimum size=1cm, inner sep=0pt}}
\tikzset{edge/.style = {-}}
  \Vertex[x=2,y=3.5]{A}
  \Vertex[x=0,y=4]{B}
  \Vertex[x=-1,y=3]{C}

  \Vertex[x=-1,y=1.5]{D}
  \Vertex[x=-0.5,y=0]{E}
  \Vertex[x=1.5,y=0]{F}
  
  \Vertex[x=3,y=0]{G}
  \Vertex[x=4,y=1]{H}
  \Vertex[x=2,y=2]{I}  

  \Vertex[x=8,y=4]{J}
  \Vertex[x=5.5,y=3.5]{K}
  \Vertex[x=5.5,y=2.5]{L}
  \Vertex[x=7.5,y=2.5]{M}

  \Vertex[x=7,y=1]{N}
  \Vertex[x=5.75,y=0]{O}
  \Vertex[x=8.5,y=0]{P}

  \Vertex[x=10,y=0]{Q}
  \Vertex[x=10, y=2]{R}
  \Vertex[x=9,y=2]{S}

  \Vertex[x=4.5,y=-0.5]{T}
  \Vertex[x=4,y=-2]{U}
  \Vertex[x=6,y=-1.5]{V}
  \tikzstyle{LabelStyle}=[fill=white,sloped]

  \Edge(A)(O)
  \Edge(H)(M)
  \Edge(L)(I)
  \Edge(N)(V)
  \Edge(G)(U)
  \Edge(O)(T)
  
  \Edge(A)(B)
  \Edge(B)(C)
  \Edge(C)(A)
  
  \Edge(D)(E)
  \Edge(E)(F)
  \Edge(D)(F)
  
  \Edge(G)(H)
  \Edge(H)(I)
  \Edge(I)(G)
  
  \Edge(J)(K)
  \Edge(K)(L)
  \Edge(L)(M)
  \Edge(J)(L)
  \Edge(J)(M)
  \Edge(K)(M)

  \Edge(N)(O)
  \Edge(O)(P)
  \Edge(N)(P)
  
  \Edge(Q)(R)
  \Edge(R)(S)
  \Edge(Q)(S)

  \Edge(T)(U)
  \Edge(U)(V)
  \Edge(T)(V)
  
  \Edge(A)(E)
  \Edge(A)(I)
  \Edge(C)(D)
  \Edge(F)(I)
  
  \Edge(J)(R)
  \Edge(L)(N)
  \Edge(M)(P)
  \Edge(P)(S)
  
\end{tikzpicture}
\caption{Full graph with all interactions}
\end{subfigure}

\begin{subfigure}{0.5\textwidth}
\begin{tikzpicture}[scale=0.65, transform shape]

\tikzset{vertex/.style = {draw, circle, fill=white, minimum size=1cm, inner sep=0pt}}
\tikzset{edge/.style = {-}}
  \Vertex[x=2,y=3.5]{A}
  \Vertex[x=0,y=4]{B}
  \Vertex[x=-1,y=3]{C}

  \Vertex[x=-1,y=1.5]{D}
  \Vertex[x=-0.5,y=0]{E}
  \Vertex[x=1.5,y=0]{F}
  
  \Vertex[x=3,y=0]{G}
  \Vertex[x=4,y=1]{H}
  \Vertex[x=2,y=2]{I}  

  \Vertex[x=8,y=4]{J}
  \Vertex[x=5.5,y=3.5]{K}
  \Vertex[x=5.5,y=2.5]{L}
  \Vertex[x=7.5,y=2.5]{M}

  \Vertex[x=7,y=1]{N}
  \Vertex[x=5.75,y=0]{O}
  \Vertex[x=8.5,y=0]{P}

  \Vertex[x=10,y=0]{Q}
  \Vertex[x=10, y=2]{R}
  \Vertex[x=9,y=2]{S}

  \Vertex[x=4.5,y=-0.5]{T}
  \Vertex[x=4,y=-2]{U}
  \Vertex[x=6,y=-1.5]{V}
  \tikzstyle{LabelStyle}=[fill=white,sloped]
    
  \tikzstyle{EdgeStyle}=[]
  \Edge(A)(O)
  \Edge(H)(M)
  \Edge(L)(I)
  \Edge(N)(V)
  \Edge(G)(U)
  \Edge(O)(T)
  
  \tikzstyle{EdgeStyle}=[color=green]
  \Edge(A)(B)
  \Edge(B)(C)
  \Edge(C)(A)
  
  \Edge(D)(E)
  \Edge(E)(F)
  \Edge(D)(F)
  
  \Edge(G)(H)
  \Edge(H)(I)
  \Edge(I)(G)
  
  \Edge(J)(K)
  \Edge(K)(L)
  \Edge(L)(M)
  \Edge(J)(L)
  \Edge(J)(M)
  \Edge(K)(M)

  \Edge(N)(O)
  \Edge(O)(P)
  \Edge(N)(P)
  
  \Edge(Q)(R)
  \Edge(R)(S)
  \Edge(Q)(S)

  \Edge(T)(U)
  \Edge(U)(V)
  \Edge(T)(V)
  
  \tikzstyle{EdgeStyle}=[color=blue]
  \Edge(A)(E)
  \Edge(A)(I)
  \Edge(C)(D)
  \Edge(F)(I)
  
  \tikzstyle{EdgeStyle}=[color=red]
  \Edge(J)(R)
  \Edge(L)(N)
  \Edge(M)(P)
  \Edge(P)(S)
  
\end{tikzpicture}
\caption{Cliques and supercliques highlighted}
\end{subfigure}

\begin{subfigure}{0.5\textwidth}
\begin{tikzpicture}[scale=0.7, transform shape]
\tikzset{vertex/.style = {draw, circle, fill=white, minimum size=1cm, inner sep=0pt}}
\tikzset{edge/.style = {-}}

  \Vertex[x=0.5,y=3]{1}

  \Vertex[x=0,y=0.5]{2}
  
  \Vertex[x=3,y=1.5]{3} 

  \Vertex[x=6.5,y=3]{4}
    
  \Vertex[x=7,y=1]{5}

  \Vertex[x=9,y=2]{6}

  \Vertex[x=5,y=0]{7}
  \tikzstyle{LabelStyle}=[fill=white,sloped]
  
  \tikzstyle{EdgeStyle}=[]
  \Edge(1)(5)
  \Edge(3)(4)
  \Edge(3)(7)
  \Edge(5)(7)

  \tikzstyle{EdgeStyle}=[color=blue]
  \Edge(1)(2)
  \Edge(2)(3)
  \Edge(1)(3)

  \tikzstyle{EdgeStyle}=[color=red]
  \Edge(4)(5)
  \Edge(5)(6)
  \Edge(4)(6)
\end{tikzpicture}
\caption{Effective graph obtained from supercliques}
\end{subfigure}
\caption{The edges in seven independent cliques are highlighted in green. The edges in the two independent supercliques they form are highlighted in blue and red, respectively. Edges in (b) participating in neither a clique nor a superclique are left black. Note that supercliques may consist of cliques with differing vertex number ({\it e.g.} the red superclique); also note that the supercliques shown in (c) form a further superclique with vertex set $\{\{1, 2, 3\}, \{4, 5, 6\}, \{7\}\}$ and edges coloured in black.}
\label{multiscale-fig}
\end{figure}
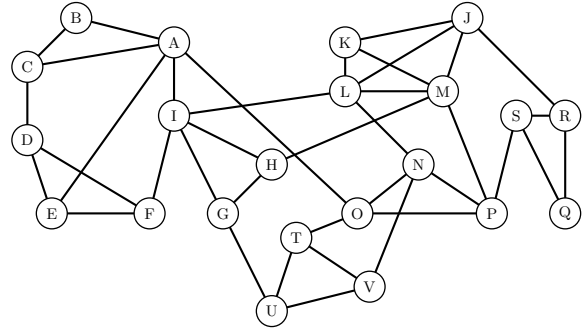
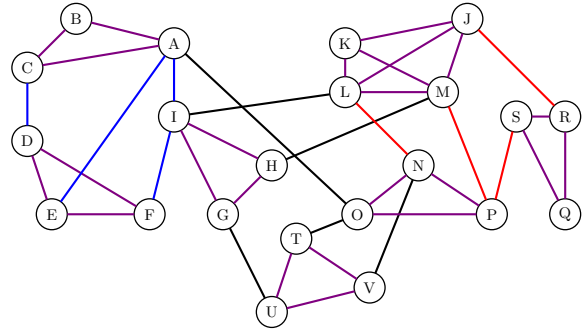
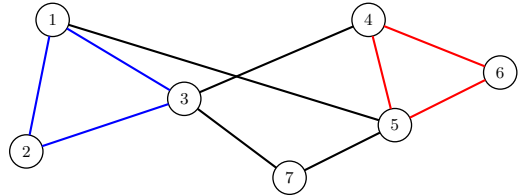

This section demonstrates that whenever we have cliques at some level, there can be interesting hybridised behaviours where local order may exist but the system may be globally disordered. This recapitulates observations of multiscale phenomena in biology and physics, where for some parameter regimes there is coherence at one level ({\it e.g.} the tissues constituting an individual organ) and non-uniformity at another ({\it e.g.} the differing organs in a body). It also shows that there can be order within a window of a transformer---if the attention heads have the right relationships---but is consistent with a lack of global coherence for all spins. 

\section{Applications and limitations}\label{limitations-section}

There are a number of compelling applications and important limitations of these results; we will discuss them here.  

Whilst a simplification, the fixed window models the practical context length limitations that real systems face. As the context increases, it becomes computationally infeasible to predict the next token---and only the last few words are fed as input to the model. Intuitively, this leads the model to completely forget what is left outside of its input window, placing constraints on how coherent the model is able to be over long periods of time. This will be a useful observation in \S\ref{sec:conclusion}. On the other hand, we have neglected to consider newly emerging strategies for expanding context or making the context window dynamic, such as those highlighted in \cite{wang2024beyond}. We hope to study these in the future. 

Theorem \ref{main-thm} is best applicable to systems falling under the Landau theory of equilibrium phase transitions. Many interesting non-equilibrium systems fall outside this regime, suggesting future directions which generalise these results to the kind of non-equilibrium free energies considered in stochastic thermodynamics \cite{qian2001relative, seifert2012stochastic, hohenberg2015introduction, seifert2018stochastic, lucarini2020response, zakine2023minimum}. However, in certain situations the same scaling relationships are respected and simply become dynamic in time or localised in space \cite{hohenberg1977theory, marro2005nonequilibrium}, meaning the results generalise readily to those cases. At criticality and out-of-equilibrium, finite-size effects become relevant to the computations of critical exponents \cite{notes}, but not to reasoning about the existence of phase transitions. Similarly, uniformity in coupling strength (as in Theorem \ref{clique-thm}) and window length (as in Theorem \ref{th:arlm}) can be assumed when studying asymptotics, since the only quantities of interest are the upper and lower bounds on the change in energy (as in Lemma \ref{lm:energy_scaling}).

\section{Conclusions and future directions}\label{sec:conclusion}

In this paper, we have developed a scaling argument extending Peierls' argument, which can be used to reason about the possibility of spontaneous order in systems with local interactions.

Biological systems are known to exploit the complex, hierarchical structures formed by cells in multicellular organisms to enable self-organisation and problem-solving in anatomical, physiological, and gene expression spaces, as well as the three-dimensional space of conventional behaviour and linguistics \cite{fields2022competency, levin2023darwin}. Many recent discussions have focussed on the similarities and differences in the dynamics that allow evolved, engineered, and hybrid systems to bind their parts towards efficient navigation to large-scale goals---a capability which is ultimately a defining feature of ``life'' \cite{mcshea2013machine, bongard2021living, watson2022design, barwich2024rage, zhang2025classical}.

\begin{figure}[h]
    \centering
    \includegraphics[width=\linewidth]{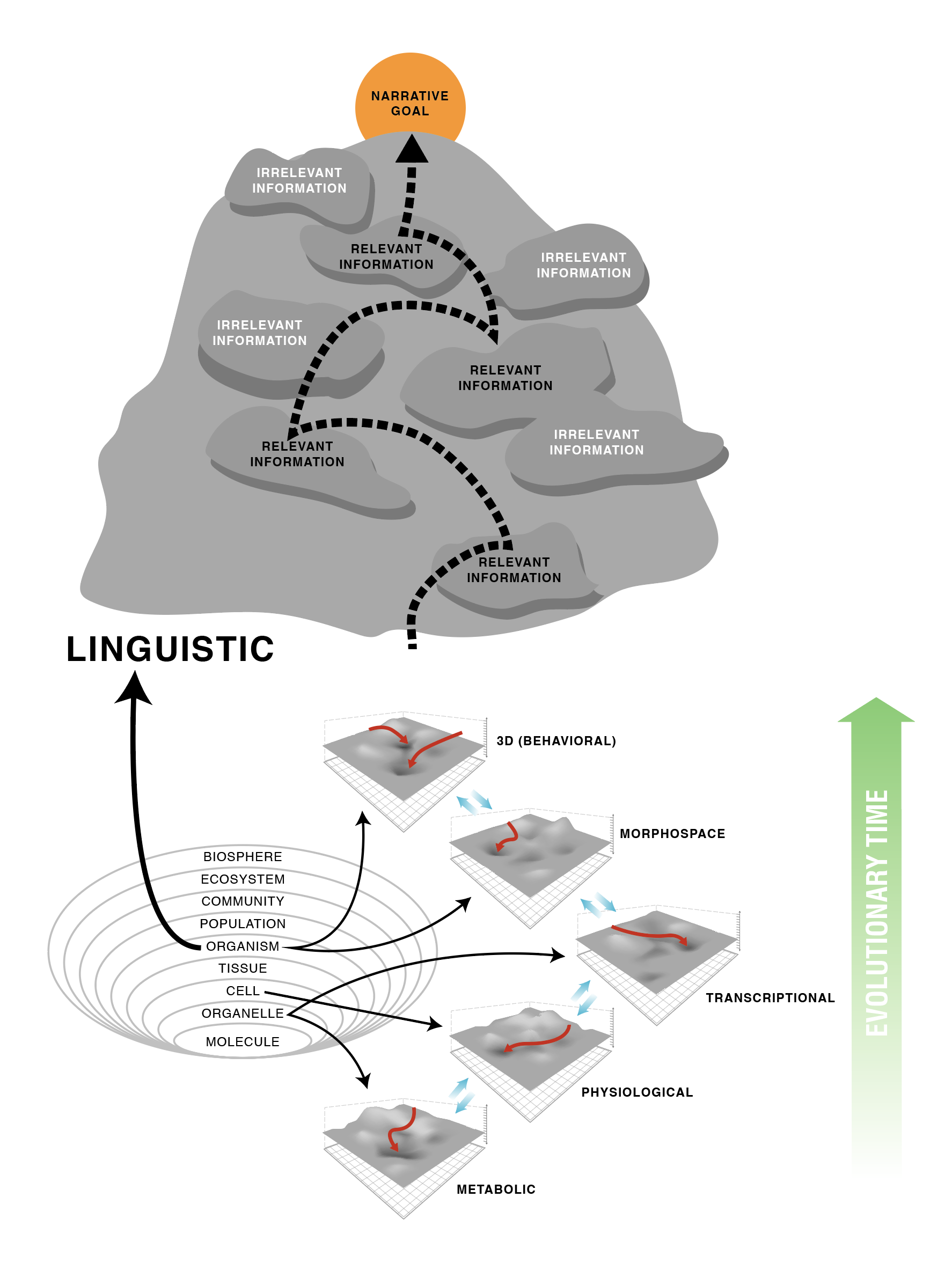}
    \caption{Cognition as goal-directed behaviour, with a focus on producing and processing meaningful pieces of text.}
    \label{fig:narrative-goals}
\end{figure}

In stark contrast, we have demonstrated that current natural language processing algorithms lack the necessary topology for self-organisation. This limitation explains their inability to generate long and coherent text that matches the complexity and consistency of biological systems. The key insight of our results is that topology is the critical factor differentiating these systems. Whilst cells in the human body can coordinate and organise over large scales, forming coherent tissues and organs, language models struggle to maintain consistency beyond their limited context windows. This disparity stems directly from the underlying topology of interactions in these systems.

By understanding the crucial role of topology in enabling self-organisation, we can begin to envision new architectures for natural language processing and biologically-inspired computing that might better emulate the self-organising capabilities of biological systems. The inability for autoregressive large language models to maintain states of long-range order resembles the {\it tangential speech} or {\it derailment} in formal thought disorder, such as that found in schizophrenia and other forms of psychosis \cite{kircher2018formal}, lending a complementary aspect to the `hallucinations' seen in current large language models. Proposed therapeutic approaches to formal thought disorder involve refining the articulated thought in conversation or written ({\it i.e.} diagrammatic) form under the guidance of a psychotherapist \cite{palmier2017cognitive}, suggesting that interacting with an environment is a way to enforce coherence of information.\footnote{We thank Karl Friston for elaborating on these points.} This further suggests that an {\it embodied} world model, extending the system in space and time by its interactions with an environment, can be leveraged to maintain coherence. We hypothesise this explains why stigmergy and other forms of extracellular signalling arise in biological systems, which is known to enhance the ability for a collective system to order itself \cite{theraulaz1999brief, marsh2008stigmergic, giuggioli2013stigmergy, gloag2015bacterial, heylighen2016stigmergy, sims2023stigmergic} and has been used in robot design and control for this feature \cite{holland1999stigmergy, valckenaers2007mas, boldini2024stigmergy}. Indeed, human brains are also unable to maintain the coherence of generated text past certain lengths of time due to energy constraints on neural processing, but have tools to cope with this, such as memory aids and cues in the environment. 

Throughout we have assumed that the free energy is minimised. For some systems the equilibration time is sufficiently long that in the interest of practicality one would be interested in relaxing this assumption---for example, when the energy landscape is very rugged. One calls such systems spin-glasses \cite{castellani2005spin}.

Studying the phase transitions of glassy systems can be challenging even for relatively simple glasses, such as Hopfield networks \cite{amit1987statistical,mezard2017mean} or the Edwards--Anderson model \cite{edwards1975theory}. In forthcoming work we will give similar estimates constraining the existence of phases with frozen disorder by the topology of the underlying lattice. 

\bibliographystyle{unsrt}
\bibliography{main}
\end{document}